\documentclass{amsart}
\usepackage{amsmath}
\usepackage{amsfonts}
\usepackage{graphicx}
\usepackage{amssymb}

\newcommand{\Z}{\mathbb{Z}}

\newcommand{\R}{\mathbb{R}}
\newcommand{\Par}{\mathcal{P}}

\newcommand{\rmv}[1]{}

\newcommand{\thedate}{\today}

\makeatletter
\newcommand{\tpitchfork}{%
  \vbox{
    \baselineskip\z@skip
    \lineskip-.52ex
    \lineskiplimit\maxdimen
    \m@th
    \ialign{##\crcr\hidewidth\smash{$-$}\hidewidth\crcr$\pitchfork$\crcr}
  }%
}
\makeatother

\newtheorem{theorem}{Theorem}[section]
\newtheorem{lemma}[theorem]{Lemma}
\newtheorem{proposition}[theorem]{Proposition}

\theoremstyle{definition}
\newtheorem{definition}[theorem]{Definition}
\newtheorem{remark}[theorem]{Remark}

\newtheorem{notation}[theorem]{Notation}

\begin{document}

\hfill  \thedate

\title[Qubo model for the Closest Vector Problem]{Qubo model for the Closest Vector Problem}


\author{Eduardo Canale}
\address{Instituto de Matem\'atica y Estad{\'\i}stica\\
Universidad de la Rep{\'u}blica, Uruguay}
\email{canale@fing.edu.uy}

\author{Claudio Qureshi}
\address{Instituto de Matem\'atica y Estad{\'\i}stica\\
Universidad de la Rep{\'u}blica, Uruguay}
\email{cqureshi@fing.edu.uy}

\author{Alfredo Viola}
\address{Instituto de Computaci{\'o}n\\
Universidad de la Rep{\'u}blica, Uruguay}
\email{viola@fing.edu.uy}

\maketitle


\begin{abstract}
In this paper we consider the closest vector problem (CVP) for lattices $\Lambda \subseteq \Z^n$ given by a generator matrix $A\in \mathcal{M}_{n\times n}(\mathbb{Z})$. Let $b>0$ be the maximum of the absolute values of the entries of the matrix $A$. We prove that the CVP can be reduced in polynomial time to a quadratic unconstrained binary optimization (QUBO) problem in $O(n^2(\log(n)+\log(b)))$ binary variables, where the length of the coefficients in the corresponding quadratic form is $O(n(\log(n)+\log(b)))$.

\end{abstract}

\vspace{3mm}

{\scriptsize \textsc MATHEMATICS SUBJECT CLASSIFICATION}


\section{Introduction and notation}

Quantum annealing is a method based on adiabatic quantum computing, which is used to solve certain optimization problems that can be expressed into the QUBO model (or into the Ising model). Although QUBO problems belong to the class of NP-hard problems, for some (small) special instances quantum annealing techniques have proven to be quite efficient. In this sense some authors have shown how to transform several important optimization problems into the QUBO framework. Among these problems are the number partitioning problem, the max-cut problem, the set packing problem and the Max 2-sat problem \cite{GKHD22i,GKHD22ii,KN98}. Motivated by the importance of lattice problems in some post-quantum cryptographic algorithms, we study how to reduce the closest vector problem (CVP) into a QUBO problem in such a way that the reduction and also the dimension and the size of the QUBO problem are polynomial in the parameters of the CVP.

The QUBO model is expressed by the optimization problem of minimizing a quadratic form $x^t Q x$ where $Q\in \mathcal{M}_{N}(\mathbb{Z})$ is a given integer square matrix and $x$ is a vector of binary variables (i.e. each coordinate $x_i \in \{0,1\}$ for $1\leq i \leq N$). We consider here two parameter of the QUBO: their dimension $N$ (i.e. the number of binary variables to be determined) and their size $s$, the maximum absolute value of the entries of the matrix $Q$. In the CVP we have as input an integer non-singular matrix $A \in \mathcal{M}_{n}(\Z)$ and a (column) vector $x\in \Z^n$ and we want to find a vector $z\in \Z^n$ which minimizes the value of $||x-Az||$. Note that all the relevant information is enclosed in the augmented matrix $\widetilde{A}=(A|x)$. Associated with the matrix $\widetilde{A}$ we define a new square matrix $Q^{\widetilde{A}}$ and our main result (Theorem \ref{thm:main}) associates to each solution of the QUBO problem $\widetilde{x}= \operatorname{argmin} x^t Q^{\widehat{A}}x$ a solution of the CVP problem $\widetilde{z}= \operatorname{argmin}_{z \in \Z^n}||x-Az||$. We also estimate the dimension and size of $Q^{\widehat{A}}$ (Remark \ref{Remark:main}). In the last section we list further interesting research problems related to the effective implementation of QUBO problem using quantum annealing and the correspondence given by our main theorem.

\begin{notation}
$\mathcal{M}_{n \times m}(\Z)$ denotes the set of $n\times m$ integer matrices, $\mathcal{M}_{n}(\Z):=\mathcal{M}_{n\times n}(\Z)$ and $\Z^n$ is the space of integer $n$-tuples which are identified with $\mathcal{M}_{n\times 1}(\Z)$ (by abuse of notation we write the vectors $x\in \Z^n$ as row vectors). The logarithm in this paper are always in base $2$. The distance between a vector $x\in \Z^n$ and a subset $\Lambda \subseteq \Z^n$ is $d(x,\Lambda):= \min_{\lambda \in \Lambda}d(x,\lambda)$.
\end{notation}


\section{A Qubo model for the CVP} 


We consider a matrix $A \in \mathcal{M}_{n}(\Z)$ with $\det(A)\neq 0$ and let $\Lambda$ be the lattice generated by $A$ (i.e. $\Lambda=\Z a_1 + \ldots + \Z a_n$ where $a_i\in \Z^{n}$ is the i-th column of $A$). We denote the Euclidean distance between $x,y \in \R^n$ by $d(x,y)=||x-y||$. Let $\mathcal{P}_{A}=\{Ax: x\in [0,1)^n\}$ be the fundamental parallelepiped of $\Lambda$ regarding the matrix $A$. Note that the problem of finding the closest lattice vector from a vector $x=\sum_{i=1}^n x_i a_i \in \R^n$ is equivalent to find the closest lattice vector from $\widehat{x}:=\sum_{i=1}^n (x_i-\left\lfloor x_i \right\rfloor) a_i \in \Par_{A}$. Indeed, if $\widehat{\lambda}\in \Lambda$ is such that $d(\widehat{x},\widehat{\lambda})=d(\widehat{x},\Lambda)$ and $\lambda_0 = \sum_{i=1}^n \left\lfloor x_i \right\rfloor a_i $ then $d(x,\widehat{\lambda}+\lambda_0)= d(\widehat{x},\widehat{\lambda})=d(\widehat{x},\Lambda)= d(\widehat{x}+\lambda_0,\Lambda+\lambda_0)=d(x,\Lambda)$. Thus we can assume that the vector $x$, from which we want to find the closest lattice vector, is in the fundamental parallelepiped $\Par_{A}$.


In the matrix space $\mathcal{M}_{n}(\R)$ we consider the spectral norm given by $||A||=\max_{x:||x||=1} ||Ax||$ and the condition number of an invertible matrix is given by $\kappa(A)=||A||\cdot ||A^{-1}||$. The following upper bound for $\kappa(A)$ is proved in \cite{GEJ95}.

\begin{lemma}\label{lemma:condition-number}
Let $A=(a_{ij}) \in \mathcal{M}_{n}(\R)$ and $||A||_{F}=\left(\sum_{i,j}a_{ij}^2\right)^{\frac{1}{2}}$ be the Frobenius norm of $A$, then $$\kappa(A)\leq \frac{2}{|\det(A)|}\cdot \left(\frac{||A||_{F}}{\sqrt{n}} \right)^n. $$
\end{lemma}

In order to obtain the closest lattice vector from a vector in the fundamental parallelepiped we can restrict to some special lattice vectors.\\

\begin{definition}
With the same notation as above, an $A$-{\it feasible} lattice vector is a lattice vector $\lambda \in \Lambda$ such that $d(\lambda, x)=d(\Lambda,x)$ for some $x\in \Par_{A}$. The set of $A$-feasible lattice vectors is denoted by $\mathcal{F}_A$.
\end{definition}

The next lemma brings an upper bound for the coordinates of feasible lattice vectors.

\begin{lemma}\label{lemma:bound-for-z}
If $\lambda=Az \in \mathcal{F}_{A}$ with $z\in \Z^n$ then $||z||\leq \left(\frac{\kappa(A)}{2} +1 \right) \cdot \sqrt{n}$.
\end{lemma}

\begin{proof}
Let $\epsilon=(\epsilon_1,\ldots, \epsilon_n) \in [0,1)^n$ such that $d(\Lambda, A \epsilon)=||\lambda - A\epsilon||$. Consider $z'=(z'_1,\ldots,z'_n)$ given by $z'_i=\lfloor \epsilon_i + \frac{1}{2} \rfloor$, $1\leq i \leq n$. We have $||z||=||A^{-1}(\lambda-A\epsilon) + \epsilon ||\leq ||A^{-1}||\cdot ||Az'-A\epsilon||+||\epsilon|| \leq \kappa(A)\cdot ||z'-\epsilon||+||\epsilon||\leq \kappa(A)\cdot \frac{\sqrt{n}}{2}+\sqrt{n}$.
\end{proof}

\begin{proposition}\label{prop:zi-bound}
Let $A=(a_{ij}) \in \mathcal{M}_{n}(\Z)$ with $n\geq 3$ and $\det(A)\neq 0$. Let $\Lambda$ be the lattice generated by the columns of $A$ and $\mathcal{F}_{A}$ its subset of $A$-feasible lattice vectors. Let $b=\max \{|a_{ij}|: 1\leq i,j\leq n\}$. If $\lambda= A z \in \mathcal{F}_{A}$ with $z=(z_1,\ldots,z_n)\in \Z^n$ then $\log(|z_i|)\leq n (\log(n)+ \log(b))$ for $1\leq i \leq n$.
\end{proposition}

\begin{proof}
By Lemma \ref{lemma:condition-number} and the bound $||A||_{F} \leq n b$ we have $\kappa(A) \leq 2n^{n/2}b^n$. By Lemma \ref{lemma:bound-for-z} we have $||z||\leq \frac{3}{2}\cdot \kappa(A)\cdot \sqrt{n}\leq 3 n^{\frac{n+1}{2}}b^n$. Then, $\log(|z_i|)\leq \log(||z||)\leq \log(3)+n\log(b)+\left(\frac{n+1}{2}\right)\cdot \log(n)\leq n(\log(n)+\log(b))$.  
\end{proof}

Let $A\in \mathcal{M}_{n\times n}(\Z)$ and $x\in \Z^n$. We denote the augmented matrix of $A$ by $\widetilde{A}=(A|x)$. Next, we introduce the Q-matrix of $\widetilde{A}$ which play an important role to reduce the CVP to a QUBO problem.

\begin{definition}\label{Def:Qmatrix}
Let $A=(a_{ij})\in \mathcal{M}_{n}(\Z)$, $x\in \mathbb{Z}^n$ and $\widetilde{A}=(A|x)$. Let $b=\max \{|a_{ij}|: 1\leq i,j\leq n\}$ and $m=\lceil n (\log(n)+ \log(b)) \rceil$. The Q-matrix of $\widetilde{A}$ is the matrix $Q^{\widetilde{A}}=(q_{hr})\in \mathcal{M}_{n(m+1)}(\Z)$ given by:
\begin{equation*}
q_{hr}:= \left\{ \begin{array}{ll}
2^{j+l} \sum_{t=1}^n a_{ti} a_{tk} & \textrm{if } h\neq r; \\
4^l \sum_{t=1}^n a_{ti}^2 - 2^{m+j+1} \sum_{t=1}^n \sum_{s=1}^n a_{st}a_{si} + 2^{j+1}\sum_{t=1}^n x_t a_{ti} & \textrm{if } h=r,
\end{array}   \right.
\end{equation*}
where $h=jn+i$, $r=ln+k$ for $0\leq j \leq m$, $1\leq i \leq n$, $0\leq l \leq m$ and $1\leq k \leq n$.
\end{definition}

Note that every positive integer $h\leq n(m+1)$ can be written univocally as $h=jn+i$ with $0\leq j \leq m$ and $1\leq i \leq n$, then the matrix $Q^{\widetilde{A}}$ is well defined.\\

Let $A,n$ and $b$ as in the Proposition \ref{prop:zi-bound}. We denote by $\mathcal{B}_{n\times m}$ the space of $n\times m$ binary matrices (i.e. each coefficient of the matrix is $0$ or $1$), $m := \lceil n (\log(n)+\log(b)) \rceil$, $u := (2^m,2^m,\ldots, 2^m)\in \Z^n$ and $p :=(1,2,\ldots,2^m)\in \Z^{m+1}$. We also consider a fixed vector $x\in \mathcal{P}_{A}$. If $z\in \Z^n$ satisfies $||Az-x||=d(\Lambda, x)$, by Proposition \ref{prop:zi-bound} we have $|z_i|\leq 2^m$ for $1\leq i \leq n$. Then, there is a binary matrix $B \in \mathcal{B}_{n\times (m+1)}$ such that $z+u = Bp$. The CVP is the problem of finding a vector $z\in \Z^n$ that minimizes $||Az-x||^2 = (Az-x)^t (Az-x)$. After a straightforward calculation substituting $z$ by $Bp-u$ and eliminating constant terms, we can see that CVP is also equivalent to the problem of finding a binary matrix $B \in \mathcal{B}_{n\times (m+1)}$ that minimizes the function: 
\begin{equation}\label{eq:binary-matrix}
q(B):= (ABp)^t ABp - 2 u^tA^tABp +2 x^tA Bp
\end{equation}
From this, it is easy to see that $q(B)$ is a quadratic polynomial in the entries of the matrix $B$, that is $q(B)= \sum c_{i,j,k,l}\cdot b_{i,j}b_{k,l} -2 \sum d_{i,j}\cdot b_{i,j}+2 \sum e_{i,j}\cdot b_{i,j}$ for some integer coefficients $c_{i,j,k,l}$, $d_{i,j}$ and $e_{i,j}$ with $1\leq i\leq n, 0\leq j \leq m, 1\leq k \leq n$ and $0\leq l \leq m$. To put it as in the standard qubo model we use the relation $b_{i,j} = b_{i,j}^2$ and rewrite the expression as $q(B)=\sum' c_{i,j,k,l}\cdot b_{i,j}b_{k,l} + \sum (c_{i,j,i,j}-2 d_{i,j} + 2 e_{i,j})\cdot b_{i,j}^2$ where $\sum'$ indicates that the sum is restricted to the $i,j,k,l$ such that $(i,j)\neq (k,l)$. In this way $q(B)= v_{B}^t Q v_{B}$ where $v_{B}$ is a binary (column) vector of length $n(m+1)$ whose $(jn+i)$-th entry is $b_{i,j}$ for $1\leq i \leq n$ and $0\leq j \leq m$ and $Q$ is a $n(m+1)\times n(m+1)$ matrix with integer coefficients $c_{i,j,k,l}$ outside the main diagonal and $c_{i,j,i,j}-2d_{i,j}+2e_{i,j}$ in the main diagonal. 

We can use Equation (\ref{eq:binary-matrix}) to calculate these coefficients and we will obtain the same values of those in Definition \ref{Def:Qmatrix}. All the above can be summarized in the following result.

\begin{theorem}\label{thm:main}
Let $A \in \mathcal{M}_{n}(\Z)$ be a non-singular matrix with $n\geq 3$ and $\Lambda_{A}$ be the lattice generated by the columns of $A$. Let $x\in \Z^n$ and $\widetilde{A}=(A|x)$. Assume that $x$ is in the fundamental parallelepiped of $\Lambda$. Let $b=\max \{|a_{ij}|: 1\leq i,j\leq n\}$, $m=\lceil n (\log(n)+ \log(b)) \rceil$, $u = (2^m,2^m,\ldots, 2^m)\in \Z^n$, $p =(1,2,\ldots,2^m)\in \Z^{m+1}$ and $\widetilde{Q}$ be the Q-matrix of $\widetilde{A}$. Consider a solution $\widetilde{x} \in \{0,1\}^{n(m+1)}$ of the QUBO problem regarding the quadratic form $x^t Q x$ (i.e $\widetilde{x}^t \widetilde{Q} \widetilde{x} \leq x^t \widetilde{Q} x$, for all $x \in \{0,1\}^{n(m+1)}$). Let $B \in \mathcal{B}_{n\times (m+1)}$ be the binary matrix with entries $B_{ij}= \widehat{x}_{jn+i}$ for $1\leq i \leq n$ and $0\leq j \leq m$. If $\widetilde{\lambda}=A(Bp-u)$ then $\widetilde{\lambda} \in \Lambda_{A}$ and $||x-\widetilde{\lambda}||\leq ||x-\lambda||$ for all $\lambda \in \Lambda_{A}$. 
\end{theorem}

\begin{remark}\label{Remark:main}
The theorem above shows that the CVP can be reduced to a QUBO problem of minimizing a quadratic form $x^t \widetilde{Q} x$ where $x$ is a vector of binary variables. The number of binary variables is $n(m+1)=O(n^2(\log(n)+\log(b)))$. The length (in bits) of the coefficients in $\widetilde{Q}$ can be bounded using the expression in Definition \ref{Def:Qmatrix} and the fact $|x_i|\leq nb$ (since $x\in \mathcal{P}_{A}$). Namely, $|q_{hr}|\leq 4^m n b^2 + 2^{2m+1}n^2 b^2 + 2^{m+1}nb^2 \leq 3\cdot 2^{2m+1}n^2b^2$. Then $\log|q_{hr}|\leq \log(3)+2m+1+2(\log(n)+\log(b))=O(n(\log(n)+\log(b)))$.
\end{remark}

\section{Further remarks}

In this paper we explicitly show how to reduce the CVP to a QUBO problem, obtaining an upper bound for the number of binary variables involved and also for the length (in bits) of the coefficients of the quadratic form $\widetilde{Q}$. Despite the fact that QUBO problems are difficult to solve in general (they are NP-hard optimization problems \cite{GKHD22i}), there are particular instances where it can be solved in a reasonable amount of time using adiabatic quantum optimization \cite{Lucas14}. The efficiency of solving a QUBO problem in a quantum annealer device depends strongly on the way the data of the matrix $\widetilde{Q}$ is processed. For example, in a D-wave computer the network topology is modeled by a host graph (the Chimera graph for a D-Wave's 2000 qubit and the Pegasus graph for a D-Wave's 5000 qubit) and the data of the matrix $\widetilde{Q}$ is given by their associated connectivity graph. A crucial point is the design of efficient embedding algorithms that are able to embed a large class of connectivity graphs (associated with QUBO problems) into the host graph of a quantum annealer device, see for instance \cite{ZBDE20}. Finding a minor-embedding is NP-complete in general, but there are efficient algorithms for some classes of graph. A continuation of this work includes to identify some particular classes of instances for the CVP such that the corresponding QUBO problem is suitable to be solved by a D-Wave computer (this will ultimately depend on the properties of the connectivity graph associated with the Q-matrix). It is important to remark that when a QUBO problem is solved via quantum annealing in general we obtain an approximate solution of the problem. Theorem \ref{thm:main} associates every exact solution $\widehat{x}$ of a QUBO problem with an exact solution $\widetilde{\lambda}$ of the corresponding CVP. It would be interesting to estimate how far we are from an optimal solution of the CVP if we start with an approximate solution $\widehat{x}$ instead of an exact solution.

\section*{Acknowledgement}
Research of the authors was supported in part by CyTeD (``Programa Iberoamericano de Ciencia y Tecnolog{\'\i}a para el Desarrollo'') project 522RT0131.

\end{document}